\documentclass[runningheads, envcountsame, a4paper]{llncs}

\usepackage{amsmath}
\usepackage{amssymb}

\renewcommand{\log}{\lg}
\newcommand{\Bleft}{B_\mathit{left}}
\newcommand{\Bright}{B_\mathit{right}}

\usepackage[ruled,vlined,linesnumbered]{algorithm2e}
\SetKwProg{Fn}{Proc}{}{}
\SetKw{KwAnd}{and}
\DontPrintSemicolon

\title{A Self-Index on Block Trees
\thanks{Funded in part by Fondecyt Grant 1-170048.}}

\author{Gonzalo Navarro}

\institute{Department of Computer Science, University of Chile, Beauchef 851, Santiago, Chile, gnavarro@dcc.uchile.cl}

\titlerunning{A Self-Index on Block Trees}
\authorrunning{G. Navarro}

\toctitle{A self-index on block trees}
\tocauthor{G. Navarro}

\date{}

\begin{document}

\maketitle

\begin{abstract}
The Block Tree is a recently proposed data structure that reaches compression
close to Lempel-Ziv while supporting efficient direct access to text substrings.
In this paper we show how a self-index can be built on top of a Block Tree
so that it provides efficient pattern searches while using space proportional
to that of the original data structure. More precisely, if a Lempel-Ziv parse
cuts a text of length $n$ into $z$ non-overlapping phrases, then
our index uses $O(z\log(n/z))$ words and finds the $occ$ occurrences of a 
pattern of length $m$ in time $O(m\log n+occ\log^\epsilon n)$ for any
constant $\epsilon>0$.


\end{abstract}

\section{Introduction}

The Block Tree (BT) \cite{BGGKOPT15} is a novel data structure for representing
a sequence, which reaches a space close to its LZ77-compressed \cite{ZL77} 
space. Given a string $S[1..n]$ over alphabet $[1..\sigma]$, on which the LZ77 
parser produces $z$ phrases (and thus an LZ77 compressor uses $z\log n + 
O(z\log\sigma)$ bits, where $\log$ denotes the logarithm in base 2), 
the BT on $S$ uses $O(z \log(n/z)\log n)$ bits (also said to be $O(z\log(n/z))$
space). This is also the best asymptotic space obtained with grammar compressors
\cite{Ryt03,CLLPPSS05,Sak05,Jez15,Jez16}.
In exchange for using more space than LZ77 compression, the BT offers fast
extraction of substrings: a substring of length $\ell$ can be extracted in
time $O((1+\ell/\log_\sigma n)\log(n/z))$. In this paper we consider
the LZ77 variant where sources and phrases do not overlap, thus 
$z=\Omega(\log n)$.

Kreft and Navarro \cite{KN13} introduced a {\em self-index} based on LZ77 
compression, which proved to be extremely space-efficient on highly repetitive
text collections \cite{CFMN16}. A self-index on $S$ is a data structure that 
offers direct access to any substring of $S$ (and thus it replaces $S$), and 
at the same time offers indexed searches. Their self-index uses $3z\log n + 
O(z\log\sigma)+o(n)$ bits (that is, about 3 times the size of the compressed 
text) and finds all the $occ$ occurrences of a pattern of length $m$ in time 
$O(m^2 h + (m+occ)\log z)$, where $h \le z$ is the maximum number of 
times a symbol is successively copied along the LZ77 parsing. A 
string of length $\ell$ is extracted in $O(h\ell)$ time.

Experiments on repetitive text collections \cite{KN13,CFMN16} show that this 
LZ77-index is smaller than any other alternative and is competitive when 
searching
for patterns, especially on the short ones where the term $m^2 h$ is small and
$occ$ is large, so that the low time to report each occurrence dominates. On
longer patterns, however, the index is significantly slower. The term $h$ can
reach the hundreds on repetitive collections, and thus it poses a significant
penalty (and a poor worst-case bound).

In this paper we design the {\em BT-index}, a self-index that builds on top
of BTs instead of on LZ77 compression. Given a BT of $w = O(z\log(n/z))$ leaves
(which can be represented in $w\log n + O(w)$ bits), the BT-index uses 
$3w\log n + O(w)$ bits, and it searches for a pattern of length $m$ in time 
$O(m^2\lg(n/z)\lg\lg z + m\lg z\lg\lg z + occ(\lg(n/z)\lg\lg n+\lg z))$, which 
is in general a better theoretical bound than that of the LZ77-index. If
we allow the space to be any $O(w)=O(z\log(n/z))$ words, then the time can be 
reduced to $O(m^2\log(n/z)+m\log^\epsilon z + occ(\log\log n+\log^\epsilon z))$
for any constant $\epsilon>0$. In regular texts, the $O(\log(n/z))$
factor is around 3--4, and it raises to 8--10 on highly repetitive texts; 
both are much lower than the typical values of $h$. Thus we expect the BT-index
to be faster than the LZ77-index especially for longer patterns, where the
$O(m^2)$ factor dominates.

The self-indexes that build on grammar compression \cite{CNfi10,CNspire12} can
use the same asymptotic space of our BT-index, and their best search time is 
$O(m^2\log\log n + m\lg z+occ\lg z)$.
Belazzougui et al.~\cite{BGGKOPT15}, however, show that in practice BTs 
are faster to access $S$ than grammar-compressed representations,
and use about the same space if the text is highly repetitive. Thus we expect
that our self-index will be better in practice than those based on grammar
compression, again especially when the pattern is long and there are no
too many occurrences to report.

There are various other indexes in the literature using $O(z\log(n/z))$ bits
\cite{GGKNP14,BEGV17} or slightly more \cite{GGKNP12,NIIBT15,BEGV17} that
offer better time complexities. However, they have not been implemented as far
as we know, and it is difficult to predict how will they behave in practice.

\section{Block Trees}

Given a string $S[1..n]$ over an alphabet $[1..\sigma]$, whose LZ77 parse 
produces $z$ phrases, a Block Tree (BT) is defined as follows. 
At the top level, numbered $l=0$, we 
split $S$ into $z$ blocks of length $b_0=n/z$. Each block is then recursively
split into two, so that if $b_l$ is the length of the blocks at level $l$ it 
holds $b_{l+1}=b_l/2$, until reaching blocks of one symbol 
after $\lg (n/z)$ levels. At each level, every pair of consecutive blocks 
$S[i..j]$ that does not appear earlier as a substring of $S[1..i-1]$ is
{\em marked}. Blocks that are not marked are replaced by a pointer $ptr$
to their first occurrence in $S$ (which, by definition, must be a marked block
or overlap a pair of marked blocks). For every level $l \ge 0$, a bitvector 
$D_l$ with one bit per block sets to $1$ the positions of marked blocks. In
level $l+1$ we consider and subdivide only the blocks that were marked in level
$l$. In this paper, this subdivision is carried out up to the last level, where
the marked blocks store their corresponding symbol. 

We can regard the BT as a binary tree (with the first $\lg z$ levels
chopped out), where the internal nodes are the marked nodes and have two
children, and the leaves are the unmarked nodes. Thus we store one pointer 
$ptr$ per leaf. We also spend one bit per node in the bitvectors $D_l$.
If we call $w$ the number of unmarked blocks (leaves), then the BT has 
$w-z$ marked blocks (internal nodes), and it uses $w\lg n + O(w)$ bits.

To extract a single symbol $S[i]$, we see if $i$ is in a marked block at
level $0$, that is, if $D_0[\lceil i/b_0\rceil]=1$. If so, we map $i$ to a 
position in the next level, which only contains the marked blocks of this level:
$$i \leftarrow (rank_1(D_0,\lceil i/b_0\rceil) -1)\cdot b_0 + 
((i-1) \!\!\mod b_0)+1.$$ 
Function $rank_c(D,p)$ counts the number of occurrences
of bit $c$ in $D[1..p]$. A bitvector $D$ can be represented in $|D|+o(|D|)$
bits so that $rank_c$ can be computed in constant time \cite{Cla96}. Therefore,
if $i$ falls in a marked block, we translate the problem to the next level in
constant time. If, instead, $i$ is not in a marked block, we take the pointer
$ptr$ stored for that block, and replace $i \leftarrow i-ptr$, assuming $ptr$
stores the distance towards the first occurrence of the unmarked block. Now $i$
is again on a marked block, and we can move on to the next level as described.
The total time to extract a symbol is then $O(\log(n/z))$.

\section{A Self-Index}

Our self-index structure is made up of two main components: the first finds
all the pattern positions that cross block boundaries, whereas the second finds
the positions that are copied onto unmarked blocks. The main property that we
exploit is the following. We will say that a block is {\em explicit} in level 
$l$ if all the blocks containing it in lower levels are marked. Note that the 
explicit blocks in level $l$ are either marked or unmarked, and the descendants
of those unmarked are not explicit in higher levels.

\begin{lemma}
The occurrences of a given string $P$ of length at least 2 in $S$ either 
overlap two explicit blocks at some level, or are completely inside an unmarked
block at some level.
\end{lemma}
\begin{proof}
We proceed by induction on the BT block size. Consider the level $l=0$, where
all the blocks are explicit. If the occurrence overlaps two blocks or it is 
completely inside an unmarked block, we are done. If, instead, it is completely
inside a marked block, then this block is split into two blocks that are 
explicit in the next level. Consider that we concatenate all the explicit blocks
of the next level. Then we have a new sequence where the occurrence appears, and
we use a smaller block size, so by the inductive hypothesis, the property holds.
The base case is the leaf level, where the blocks are of length 1.
\qed
\end{proof}

We exploit the lemma in the following way. We will define an occurrence of $P$
as {\em primary} if it overlaps two consecutive blocks at some level. The
occurrences that are completely contained in an unmarked block are {\em
secondary} (this idea is a variant of the classical one used in all the
LZ-based indexes \cite{KU96}). Secondary occurrences are found by detecting
primary or other secondary occurrences within the area from where an unmarked
block is copied. We will use a data structure to find the primary occurrences 
and another to detect the copies.

\begin{lemma}
The described method correctly identifies all the occurrences of a string
$P$ in $S$.
\end{lemma}
\begin{proof}
We proceed again by induction on the block length. Consider level $l=0$. If
a given occurrence overlaps two explicit blocks at this level, then it is
primary and will be found. Otherwise, if it is inside a marked block at this
level, then it also appears at the next level and it will be found by the 
inductive hypothesis. Finally, if it is inside an unmarked block, then it
points to a marked block at the same level and will be detected as a copy of
the occurrence already found in the source. The base case is the last level,
where all the blocks are of length 1.
\qed
\end{proof}

\subsection{The Data Strucures}

We describe the data structures used by our index. Overall, they require 
$3w\lg n + O(w)$ bits, and replace the pointers $ptr$ used by the
original structure. We also retain the bitvectors $D_l$, which add up to
$O(w)$ bits.

\paragraph{Primary occurrences.}
Our structure to find the primary occurrences is a two-dimensional discrete 
grid $G$ storing points $(x,y)$ as follows. Let $B$ be a marked block at some
level $l$, which is divided into $B = \Bleft \cdot \Bright$ at level $l+1$.
Note that blocks $\Bleft$ and $\Bright$ can be marked or unmarked. Then we 
collect the reverse block $\Bleft^{rev}$ (i.e., $\Bleft$ read backwards)
in the multiset $Y$ and the block $\Bright$ in the multiset $X$. In addition,
for the blocks $B_1 \ldots B_z$ of level $l=0$, we also add $B_i^{rev}$ to $Y$
and the suffix $B_{i+1} \ldots B_z$ to $X$, for all $1 \le i < z$.

We then lexicographically sort $X$ and $Y$, to obtain the strings $X_1, X_2, 
\ldots$ and $Y_1, Y_2, \ldots$. The grid has a point at $(x,y)$ for each pair
$\langle X_x, Y_y\rangle$ that was added together in the previous process.

If a primary occurrence is not contained in any block, then it spans a sequence
$B_i \ldots B_j$ of blocks at level $l=0$. We will then find it as the 
concatenation of a suffix $B_i$ with a prefix of $B_{i+1} \ldots B_z$.
Otherwise, let $B$ be the smallest (or deepest) marked block that contains the
occurrence. Let $B$ be split into $\Bleft \cdot \Bright$ in the BT. Then the
occurrence will span a suffix of $\Bleft$ and a prefix of $\Bright$ (otherwise,
$B$ is not minimal or the occurrence is not primary). Therefore,
each primary occurrence will be found at exactly one pair $\langle X_x,Y_y
\rangle$.

The grid $G$ is of size $w \times w$, since there are $w-1$ pairs
$\langle X_x, Y_y \rangle$: one per internal BT node (of which there are $w-z$),
plus $z-1$ for the blocks of level 0.
We represent $G$ using a wavelet tree \cite{GGV03,GRR08,Nav14}, so that it takes
$w \lg w + o(w)$ bits and can report all the $y$-coordinates of the $p$
points lying inside any rectangle of the grid in time $O((p+1)\log w)$. We 
spend other $w\lg n$ bits in an array $T[1..w]$ that gives the position 
$j$ in $S$ corresponding to each point $(x,y)$, sorted by $y$-coordinate.

\paragraph{Secondary occurrences.}
Let $S_l[1..n_l]$ be the subsequence of $S$ formed by the explicit blocks at 
level $l$. If an unmarked block $B_i[1..b_l]$ at level $l$ points to its first
occurrence at $S_l[k..k+b_l-1]$, we say that $[k..k+b_l-1]$ is the {\em 
source} of $B_i$. 

For each level $l$ with $w_l$ unmarked blocks, we store two structures to find 
the secondary occurrences. The first is a bitvector $F_l[1..n_l+w_l]$ built as 
follows: We traverse from $S_l[1]$ to $S_l[n_l]$. For 
each $S_l[k]$, we add a $0$ to $F_l$, and then as many $1$s as sources start 
at position $k$. The second structure is a permutation $\pi_l$ on $[w_l]$ 
where $\pi_l(i)=j$ iff the source of the $i$th unmarked block of level $l$
is signaled by the $j$th $1$ in $F_l$.

Each bitvector $F_l$ can be represented in $w_l \lg(n_l/w_l) + O(w_l)$
bits so that operation $select_1(F_l,r)$ can be computed in constant time
\cite{OS07}. This operation finds the position of the $r$th 1 in $F_l$.
On the other hand, we represent $\pi_l$ using a structure \cite{MRRR12} that 
uses $w_l \lg w_l + O(w_l)$ bits and computes any $\pi_l(i)$ in constant time
and any $\pi_l^{-1}(j)$ in time $O(\log w_l)$. Added over all the levels, since
$\sum_l w_l = w$, these structures use $w\log n + O(w)$ bits. 


\subsection{Extraction} \label{sec:extract}

Let us describe how we extract a symbol $S[i]=S_0[i]$ using our representation.
We first compute the block $j\leftarrow \lceil i/b_0 \rceil$ where $i$ falls. 
If $D_0[j]=1$, we are already done on this level. If, instead, $D_0[j]=0$,
then the block $j$ is not marked. Its rank among the unmarked blocks of this
level is $r_0 = rank_0(D_0,j)$. The position of the $1$ in $F_0$ corresponding
to its source is $p_0 = select_1(F_0,\pi_0(r_0))$. This means that the source
of the block $j$ starts at $S_0[p_0 - \pi_0(r_0)]$. Since block $j$ starts at
position $s_0 = (j-1)\cdot b_0 + 1$, we set
$i \leftarrow (p_0-\pi_0(r_0)) + (i-s_0)$ and recompute 
$j\leftarrow \lceil i/b_0 \rceil$, 
knowing that the new symbol $S_0[i]$ is the same as the original one.

\begin{algorithm}[t]

\Fn{Extract$(i)$}{
 $l \leftarrow 0$ \\
 $b \leftarrow n/z$ \\
 \While {$b > 1$}
    { $j \leftarrow \lceil i/b \rceil$ \\
      \If {$D_l[j] = 0$}
	  { $r \leftarrow rank_0(D_l,j)$ \\
	    $p \leftarrow select_1(F_l,\pi_l(r))$ \\
	    $s \leftarrow (j-1)\cdot b + 1$ \\
	    $i \leftarrow (p-\pi_l(r)) + (i-s)$ \\
	    $j \leftarrow \lceil i/b \rceil$ \\
	  }
      $i \leftarrow (rank_1(D_l,j)-1)\cdot b + ((i-1) \!\!\mod b) + 1$ \\
      $l \leftarrow l+1$ \\
      $b \leftarrow b/2$
    }
  Return the symbol stored at position $i$ in the last level
}
\caption{Extracting symbols from our encoded BT.}
\label{alg:extract}
\end{algorithm}

Now that $i$ is inside a marked block $j$, we move to the next level. To
compute the position of $i$ in the next level, we do
$i \leftarrow (rank_1(D_0,j)-1)\cdot b_0 + ((i-1) \!\!\mod b_0) + 1$,
and continue in the same way to extract $S_1[i]$. In the last level we find 
the
symbol stored explicitly. The total time to extract a symbol is $O(\log(n/z))$.

Algorithm~\ref{alg:extract} gives the pseudocode.

\subsection{Queries}

\paragraph{Primary occurrences.}
To search for a pattern $P[1..m]$, we first find its primary occurrences using
$G$ as follows. For each partition $P_< = P[1..k]$ and $P_> = P[k+1..m]$, for
$1 \le k < m$, we binary search $Y$ for $P_<^{rev}$ and $X$ for $P_>$. To
compare $P_<^{rev}$ with a string $Y_i$, since $Y_i$ is not stored, we extract
the consecutive symbols of $S[T[i]-1]$, $S[T[i]-2]$, and so on, until the
lexicographic comparison can be decided. Thus each comparison requires
$O(m \log(n/z))$ time. To compare $P_>$ with a string $X_i$, since $X_i$ is
also not stored, we extract the only point of the range $[i,i] \times [1,w]$
(or, in terms of the wavelet tree, we extract the $y$-coordinate of the $i$th
element in the root sequence),
in time $O(\log w)$. This yields the point $Y_j$. Then we compare $P_>$ with
the successive symbols of $S[T[j]]$, $S[T[j]+1]$, and so on. Such a comparison
then costs $O(\log w + m\log(n/z))$. The $m$ binary searches require $m\log w$
binary search steps, for a total cost of $O(m^2 \log w \log (n/z) + m\log^2 w)$.

Note that the length of the strings to compare can be obtained implicitly from
$T[i]$ (or, equivalently, $T[j]$). If $T[i]-1$ (or $T[j]-1$) is a multiple of 
$(n/z)/2^l$ but not of $(n/z)/2^{l+1}$, then the string is a block of level $l$
and its length is $(n/z)/2^l$ (except if $l=0$, in which case $Y_j$ is the
full suffix $S[T[j]..n]$). This is easily found in constant time using
arithmetic operations.

Each couple of binary searches identifies ranges $[x_1,x_2] \times [y_1,y_2]$,
inside which we extract every point. The $m$ range searches cost $O(m\log w)$ 
time. Further, each point $(x,y)$ extracted costs $O(\log w)$ and it identifies
a primary occurrence at $S[T[y]-k..T[y]-k+m-1]$. Therefore the total cost with
$occ_p$ primary occurrences is $O(m^2 \log w \log (n/z) + m\log^2 w + 
occ_p\log w)$.

Algoritm~\ref{alg:search} gives the general search procedure, using procedure
{\em Primary} to report the primary occurrences and all their associated
secondary ones.

\begin{algorithm}[t]

\Fn{Search$(P,m)$}{
  \If {$m=1$}
      { $m \leftarrow 2$ \\
	$P = P[1]*$
      }
  \For {$k=1$ \KwTo $m-1$}
       { $[x_1,x_2] \leftarrow$ binary search for $P[k+1..m]$ in $X_1,\ldots,X_w$ $~~~~~~~~~~~~~~~~~$ $~~~~~~~~~~~~~~$ (or $[1,w]$ if $P[k+1..m]=*$) \\
         $[y_1,y_2] \leftarrow$ binary search for $P[1..k]^{rev}$ in $Y_1,\ldots,Y_w$ \\
	 \For {$(x,y) \in G \cap [x_1,x_2] \times [y_1,y_2]$}
	     { \textit{Primary}$(T[y]-k,m)$
	     }
       }
}
\caption{General search procedure.}
\label{alg:search}
\end{algorithm}

Patterns $P$ of length $m=1$ can be handled as $P[1]*$, where $*$ stands for
any character. Thus we take $[x_1,x_2]=[1,w]$ and carry out the search as a 
normal pattern of length $m=2$. To make this work also for the last position 
in $S$, we assume as usual that $S$ is terminated by a special character \$.

To speed up the binary searches, we can sample one out of $\log w$ strings
from $Y$ and insert them into a Patricia tree \cite{Mor68}, which would use
$O(w)$ extra space. The up to $\sigma$ children in each node are stored in
perfect hash functions, so that in $O(m)$ time we can find the Patricia tree
node $v$ representing the pattern prefix or suffix sought. Then the range
$[y_1,y_2]$ includes all the sampled leaves descending from $v$, and up to
$\log w$ strings preceding and following the range. The search is then 
completed with binary searches in $O(\log\log w)$ steps. In case the pattern
prefix or suffix is not found in the Patricia tree, we end up in a node $v$
that does not have the desired child and we have to find the consecutive pair
of children $v_1$ and $v_2$ that surround the nonexistent child. A predecessor
search structure per node finds these children in time $O(\log\log\sigma) =
O(\log\log z) = O(\log\log w)$. Then we finish with a binary search between 
the rightmost leaf of $v_1$ and the leftmost leaf of $v_2$, also in
$O(\lg\lg w)$ steps. Each binary search step takes $O(m\lg(n/z))$ time to
read the desired substring from $S$. At the end of the Patricia search, we 
must also read one string and verify that the range is correct, but this cost 
is absorbed in the binary searches.
Overall, the search for each cut of the pattern costs $O(m\lg(n/z)\lg\lg w)$.
We proceed similarly with $X$, where there is an additional cost of 
$O(\lg w\lg\lg w)$ to find the position where to extract each string from. 
The total cost over all 
the $m-1$ searches is then $O(m(m\lg(n/z)+\lg w)\lg\lg w)$.

\paragraph{Secondary occurrences.}
Let $S[i..i+m-1]$ be a primary occurrence. This is already a range $[i_0..i_0+
m-1] = [i..i+m-1]$ at level $l=0$. We track the range down to positions 
$[i_l..i_l+m-1]$ at all the levels $l>0$, using the position tracking mechanism
described in Section~\ref{sec:extract} for the case of marked nodes: 
$$i_{l+1} = 
(rank_1(D_l,\lceil i_l/b_l\rceil)-1)\cdot b_l + ((i_l-1) \!\!\mod b_l) + 1.$$
Note that we only need to consider levels $l$ where the block length is 
$b_l \ge m$, as with shorter blocks there cannot be secondary occurrences. So 
we only consider the levels $l=0$ to $l=\lg(n/z)-\lg m$. Further, we should
ensure that the block or the two blocks where $[i_l..i_l+m-1]$ lies are marked
before projecting the range to the next level, that is,
$D_l[\lceil i_l / b_l \rceil] = D_l[\lceil (i_l+m-1) / b_l \rceil] = 1$.
Still, note that we can ignore this test, because there cannot be sources
spanning concatenated blocks that were not contiguous in the previous levels.

For each valid range $[i_l..i_l+m-1]$, we determine the sources that contain
the range, as their target will contain a secondary occurrence. Those sources
must start between positions $k = i_l+m-b_l$ and $k' = i_l$.
We find the positions $p=select_0(F_l,k)$ and $p'=select_0(F_l,k'+1)$,
thus the blocks of interest are $\pi_l^{-1}(t)$, from $t=p-k+1$ to
$t=p'-k'-1$. Since $F_l$ is represented as a sparse bitvector \cite{OS07}, 
operation $select_0$ is solved with binary search on $select_1$, in time
$O(\log w_l) = O(\log w)$. This can be accelerated to $O(\log\log n_l)$ by
sampling one out of $\log n_l$ 1s in $F_l$, building a predecessor structure
on the samples, and then completing the binary search within two samples.
The extra space of the predecessor structures adds up to $O(w)$ bits.

To report the occurrence inside each such block $q=\pi_l^{-1}(t)$, we first
find its position in the corresponding unmarked block in its level. The block
starts at $S_l[(select_0(D_l,q)-1)\cdot b_l + 1]$, and the offset of the 
occurrence inside the block is $i_l - (select_1(F_l,t)-t)$ (operation 
$select_c$ on $D_l$ is answered in constant time using $o(|D_l|)$ further
bits \cite{Cla96}). Therefore, the
copied occurrence is at $S_l[i'_l..i'_l+m-1]$, where
$$i'_l = ((select_0(D_l,q)-1)\cdot b_l + 1) + (i_l - (select_1(F_l,t)-t)).$$
We then project the position $i_l'$ upwards until reaching the
level $l=0$, where the positions correspond to those in $S$. To project 
$S_l[i_l']$ to $S_{l-1}$, we compute the block number 
$j=\lceil i_l'/b_{l-1}\rceil$, and set 
$$i_{l-1}' \leftarrow 
(select_1(D_{l-1},j)-1)\cdot b_{l-1} + ((i_l'-1) \!\!\mod b_{l-1})+1.$$

Each new secondary occurrence we report at $S[i..i+m-1]$ must be also processed
to find further secondary occurrences at unmarked blocks copying it at any 
level. This can be done during the upward tracking to find its position in $S$,
as we traverse all the relevant ranges $[i_l'..i_l'+m-1]$.

Algorithm~\ref{alg:report} describes the procedure to report the primary
occurrence $S[i..i+m-1]$ and all its associated secondary occurrences.

\begin{algorithm}[t]

\Fn{Primary$(i,m)$}{
  $l \leftarrow 0$ \\
  $b \leftarrow n/z$ \\
  \While {$b/2 \ge m$ \KwAnd $D_l[\lceil i/b\rceil]=D_l[\lceil (i+m-1)/b\rceil]=1$}
     { 
       $i \leftarrow (rank_1(D_l,\lceil i/b\rceil)-1)\cdot b + ((i-1) \!\!\mod b) + 1$ \\
       $l \leftarrow l+1$ \\
       $b \leftarrow b/2$ \\
     }
  \textit{Secondary}$(l,i,m)$
}

\Fn{Secondary$(l,i,m)$}{
 $b \leftarrow (n/z)/2^l$ \\
 \While {$l \ge 0$}
    { $k \leftarrow i + m - b$ \\
      $k' \leftarrow i$ \\
      $p \leftarrow select_0(F_l,k)$ \\
      $p' \leftarrow select_0(F_l,k')$ \\
      \For {$t \leftarrow p-k+1$ \KwTo $p'-k'-1$}
	   { $q \leftarrow \pi_l^{-1}(t)$ \\
             $i' \leftarrow ((select_0(D_l,q)-1)\cdot b + 1) + (i - (select_1(F_l,t)-t))$ \\
	     \textit{Secondary}$(l,i',m)$ \\
	   } 
      $b \leftarrow 2\cdot b$ \\
      $l \leftarrow l-1$ \\
      \If {$l \ge 0$} 
          { $j \leftarrow \lceil i/b \rceil$ \\
            $i \leftarrow (select_1(D_l,j)-1)\cdot b + ((i-1) \!\!\mod b)+1$ 
	  }
    }
  Report occurrence at position $i$ 
}
\caption{Reporting primary and secondary occurrences.}
\label{alg:report}
\end{algorithm}

Considering the time to compute $\pi_l^{-1}$ at its source, the upward
tracking to find its position in $S$, and the tests to find further secondary
occurrences at each level of the upward tracking, each secondary occurrence is 
reported in time $O(\log(n/z)\log\log n)$. Each primary occurrence, in 
turn, is obtained in time $O(\log w)$ and then we spend $O(\log(n/z)\log\log n)$
time to track it down to all the levels to find possible secondary occurrences.
Therefore, the $occ$ primary and secondary occurrences are reported in time
$O(occ(\log(n/z)\log\log n + \log w))$.

\paragraph{Total query cost.}
As described, the total query cost to report the $occ$ occurrences is 
$O(m^2 \log (n/z)\log\log w + m\log w\lg\lg w + occ (\log(n/z)\log\log n+\log w))$. 
Since $w =O(z\lg(n/z))$ and $z=\Omega(\log n)$, it holds $\lg w=\Theta(\lg z)$.
A simplified formula is $O(m^2\lg n\lg\lg z+occ\lg n\lg\lg n)$. 
The space is $3w\log n + O(w)$ bits.

\begin{theorem}
Given a string $S[1..n]$ that can be parsed into $z$ non-overlapping Lempel-Ziv
phrases and represented with a BT of $w=O(z\log(n/z))$ pointers, there 
exists a data structure using $3w\log n + O(w)$ bits that so that any substring
of length $\ell$ can be extracted in time $O(\ell \log(n/z))$ and the $occ$ 
occurrences of a pattern $P[1..m]$ can be obtained in time 
$O(m^2\lg(n/z)\lg\lg z + m\log z\lg\lg z + occ(\lg(n/z)\lg\lg n+\lg z))$. 
This can be written as $O(m^2\lg n\lg\lg z+occ\lg n\lg\lg n)$.
\end{theorem}

If we are interested in a finer space result, we can see that
the space is actually $2w\log n + w\log w + O(w)$ bits. This can be reduced 
to $w\log n + 2w\log w + O(w)$ by storing the array $T[1..w]$ in
$w\lg w + O(w)$ bits as follows. We have 
that each such position is either the start of a block at level $l=0$ or
the middle of a marked block. If we store the bitvectors $D_0$ to
$D_{\log(n/z)}$ concatenated into $D=1^z D_0 \cdots D_{\log(n/z)}$, then the
first $z$ 1s represent the blocks at level $l=0$ and the other 1s represent the
marked blocks of each level. We can therefore store $T[k]=p$ to refer to the
$p$th 1 in $D$, so that $T$ uses $w\log w$ bits. From the position
$select_1(D,p)$ in $D$, we can determine in constant time if it is among
the first $z$, which corresponds to a level-0 block, or that it corresponds to
some $D_l[i]$ (by using $rank$ on another bitvector of $O(w)$ bits that marks
the $\lg(n/z)$
starting positions of the bitvectors $D_l$ in $D$, or with a small fusion tree
storing those positions). If $T[k]$ points to $D_l[i]$, we know that the
suffix starts at $S_l[i_l]$, for $i_l = (i-1/2)\cdot b_l+1$. We then
project this position up to $S$. Thus we obtain any position of $T$ in time 
$O(\log(n/z))$, which does not affect the complexities.

\section{Using Linear Space}

If we do not care about the constant multiplying the space, we can have a
BT-index using $O(w\log n)$ bits and speed up searches considerably. First, 
we can obtain the $m-1$ ranges in the sets $X$ and $Y$, corresponding to
prefixes/suffixes of $P$, in overall time $O(m\log(mn/z))$, by using the 
fingerprinting technique described by Gagie et al.~\cite[Lem.~9]{GNP17}.
The lemma states that, if one can extract a substring of length $\ell$ from $S$
in time $f_e(\ell)$ and can compute a {\em fingerprint} of it in time 
$f_h(\ell)$, then one can obtain the lexicographic ranges in $X$ of the $m-1$ 
suffixes of $P$ in time $O(m\log(\sigma)/\omega + m(f_h(m)+\log m)+f_e(m))$, 
where $\omega = \Omega(\log n)$ is the number of bits in the computer word.
Similarly, we obtain the ranges in $Y$ of the suffixes of the reversed pattern.
They also show \cite[Thm.~2]{GNP17} how to extract any substring 
$S[i..i+\ell-1]$ in time $O(h + m\log(\sigma)/\omega)$ using a 
structure analogous to block trees of height $h$. These can be translated
verbatim to block-trees using $O(w\log n)$ bits that extract substrings
in time $O(\log(n/z)+m\log(\sigma)/\omega)$. Finally, they show
\cite[Lem.~7]{GNP17} how to compute a fingerprint in time $O(h)$ with the 
same structure analogous to a block tree of height $h$, so we can similarly
extend the block tree to use $O(w)$ words and compute the fingerprint in time
$O(\log(n/z))$.

Second, we can use faster two-dimensional range search data structures that 
still require linear space \cite{CLP11} to report the $p$ points in time 
$O((p+1)\log^\epsilon w)$ for any constant $\epsilon>0$ \cite{CLP11}. This 
reduces the cost per primary occurrence to $O(\log(n/z)\log\log n +
\log^\epsilon w)$. 

Finally, we can replace the predecessor searches that implement $select_0$
on the bitvectors $F_l$ by a completely different mechanism. Note that all
those searches we perform in our upward or downward path refer to the same
occurrence position $S[i..i+m-1]$, because we do not find unmarked blocks in 
the path. Thus, instead of looking for sources covering the occurrence at every
step in the path, we
use a single structure where all the sources from all the levels $l$ are mapped
to $S$. Such sources $[j..j+b_l-1]$ are sorted by their starting positions $j$ 
in an array $R[1..w]$. We create a range maximum query data structure 
\cite{FH11} on $R$, able to find in constant time the maximum endpoint 
$j+b_l-1$ of the blocks in any range of $R$. A predecessor search structure on
the $j$ values gives us the rightmost position $R[r]$ where the blocks start at
$i$ or to its left. A range maximum query on $R[1..r]$ then finds the block 
$R[k]$ with the rightmost endpoint in $R[1..r]$. If even $R[k]$ does not cover 
the position $j+b_l-1$, then no source covers the occurrence. If it does, we 
process it as a secondary occurrence and recurse on the ranges $R[1..k-1]$ and 
$R[k+1..r]$. It is easy to see that each valid secondary occurrence is 
identified in $O(1)$ time. 

Note that, if we store the starting position 
$j'$ of the target of source $[j..j+b_l-1]$, then we directly have the position
of the secondary occurrence in $S$, $S[i'..i'+m-1]$ with $i' = j'+(i-j)$. Thus
we do not even need to traverse paths upwards or downwards, since the primary 
occurrences already give us positions in $S$. The support for inverse
permutations $\pi_l^{-1}$ becomes unnecessary. Then the cost per secondary 
occurrence is reduced to a predecessor search.
A similar procedure is described for the LZ77-index \cite{KN13}.

The total time then becomes $O(m\lg(mn/z) + m\lg^\epsilon z + occ(\lg\lg(n/z) +
\lg^\epsilon z)) = O(m\log n + occ(\lg\lg n + \lg^\epsilon z))$.

\begin{theorem}
A string $S[1..n]$ where the LZ77 parse produces $z$ non-overlapping phrases 
can be represented in $O(z\log(n/z))$ space so that any substring of length 
$\ell$ can be extracted in time $O(\log(n/z)+\ell/\log_\sigma n)$ and the 
$occ$ occurrences of a pattern $P[1..m]$ can be obtained in time 
$O(m\log n + occ(\lg\lg n+\log^\epsilon z))$, for any constant $\epsilon>0$.
This can be written as $O(m\lg n + occ\lg^\epsilon n)$.
\end{theorem}

\section{Conclusions}

We have proposed a way to build a self-index on the Block Tree (BT)
\cite{BGGKOPT15} data structure, which we call BT-index. The BT obtains a 
compression related to the LZ77-parse of the string. If the parse uses $z$ 
non-overlapping phrases, then the BT uses $O(z\log(n/z))$ space, whereas an 
LZ77-compressor uses $O(z)$ space. Our BT-index, within the same asymptotic 
space of a BT, finds all the $occ$ occurrences of a pattern $P[1..m]$ in time 
$O(m\log n+occ\log^\epsilon n)$ for any constant $\epsilon>0$. 

The next step is to implement the BT-index, or a sensible simplification of it,
and determine how efficient it is compared to current implementations 
\cite{KN13,CNfi10,CNspire12,CFMN16}. As discussed in the Introduction, there 
are good reasons to be optimistic about the practical performance of this 
self-index, especially when searching for relatively long patterns.

\section*{Acknowledgements}

Many thanks to Simon Puglisi and an anonymous reviewer for pointing out several
fatal typos in the formulas, and to Travis Gagie for useful suggestions.

\bibliographystyle{splncs03}
\bibliography{paper}

\begin{thebibliography}{10}
\providecommand{\url}[1]{\texttt{#1}}
\providecommand{\urlprefix}{URL }

\bibitem{BGGKOPT15}
Belazzougui, D., Gagie, T., Gawrychowski, P., K{\"a}rkk{\"a}inen, J.,
  Ord{\'o\~n}ez, A., Puglisi, S.J., Tabei, Y.: Queries on {LZ}-bounded
  encodings. In: Proc. 25th Data Compression Conference (DCC). pp. 83--92
  (2015)

\bibitem{BEGV17}
Bille, P., Ettienne, M.B., G{\o}rtz, I.L., Vildh{\o}j, H.W.: Time-space
  trade-offs for {L}empel-{Z}iv compressed indexing. In: Proc. 28th Annual
  Symposium on Combinatorial Pattern Matching (CPM). pp. 16:1--16:17. LIPIcs 78
  (2017)

\bibitem{CLP11}
Chan, T.M., Larsen, K.G., P{\u{a}}tra\c{s}cu, M.: Orthogonal range searching on
  the {RAM}, revisited. In: Proc. 27th ACM Symposium on Computational Geometry
  (SoCG). pp. 1--10 (2011)

\bibitem{CLLPPSS05}
Charikar, M., Lehman, E., Liu, D., Panigrahy, R., Prabhakaran, M., Sahai, A.,
  Shelat, A.: The smallest grammar problem. IEEE Transactions on Information
  Theory  51(7),  2554--2576 (2005)

\bibitem{Cla96}
Clark, D.: Compact {PAT} Trees. Ph.D. thesis, University of Waterloo, Canada
  (1996)

\bibitem{CFMN16}
Claude, F., Fari{\~n}a, A., Mart{\'{\i}}nez-Prieto, M., Navarro, G.: Universal
  indexes for highly repetitive document collections. Information Systems  61,
  1--23 (2016)

\bibitem{CNfi10}
Claude, F., Navarro, G.: Self-indexed grammar-based compression. Fundamenta
  Informaticae  111(3),  313--337 (2010)

\bibitem{CNspire12}
Claude, F., Navarro, G.: Improved grammar-based compressed indexes. In: Proc.
  19th International Symposium on String Processing and Information Retrieval
  (SPIRE). pp. 180--192. LNCS 7608 (2012)

\bibitem{FH11}
Fischer, J., Heun, V.: Space-efficient preprocessing schemes for range minimum
  queries on static arrays. SIAM Journal on Computing  40(2),  465--492 (2011)

\bibitem{GGKNP12}
Gagie, T., Gawrychowski, P., K{\"{a}}rkk{\"{a}}inen, J., Nekrich, Y., Puglisi,
  S.J.: A faster grammar-based self-index. In: Proc. 6th International
  Conference on Language and Automata Theory and Applications (LATA). pp.
  240--251. LNCS 7183 (2012)

\bibitem{GGKNP14}
Gagie, T., Gawrychowski, P., K{\"a}rkk{\"a}inen, J., Nekrich, Y., Puglisi,
  S.J.: {LZ77}-based self-indexing with faster pattern matching. In: Proc. 11th
  Latin American Symposium on Theoretical Informatics (LATIN). pp. 731--742
  (2014)

\bibitem{GNP17}
Gagie, T., Navarro, G., Prezza, N.: Optimal-time text indexing in {BWT}-runs
  bounded space. CoRR  abs/1705.10382 (2017)

\bibitem{GRR08}
Golynski, A., Raman, R., Rao, S.S.: On the redundancy of succinct data
  structures. In: Proc. 11th Scandinavian Workshop on Algorithm Theory (SWAT).
  pp. 148--159. LNCS 5124 (2008)

\bibitem{GGV03}
Grossi, R., Gupta, A., Vitter, J.S.: High-order entropy-compressed text
  indexes. In: Proc. 14th Annual ACM-SIAM Symposium on Discrete Algorithms
  (SODA). pp. 841--850 (2003)

\bibitem{Jez15}
Jez, A.: Approximation of grammar-based compression via recompression.
  Theoretical Computer Science  592,  115--134 (2015)

\bibitem{Jez16}
Jez, A.: A really simple approximation of smallest grammar. Theoretical
  Computer Science  616,  141--150 (2016)

\bibitem{KU96}
K{\"{a}}rkk{\"{a}}inen, J., Ukkonen, E.: {Lempel-Ziv} parsing and
  sublinear-size index structures for string matching. In: Proc. 3rd South
  American Workshop on String Processing (WSP). pp. 141--155 (1996)

\bibitem{KN13}
Kreft, S., Navarro, G.: On compressing and indexing repetitive sequences.
  Theoretical Computer Science  483,  115--133 (2013)

\bibitem{Mor68}
Morrison, D.: {PATRICIA} -- practical algorithm to retrieve information coded
  in alphanumeric. Journal of the ACM  15(4),  514--534 (1968)

\bibitem{MRRR12}
Munro, J.I., Raman, R., Raman, V., Rao, S.S.: Succinct representations of
  permutations and functions. Theoretical Computer Science  438,  74--88 (2012)

\bibitem{Nav14}
Navarro, G.: Wavelet trees for all. Journal of Discrete Algorithms  25,  2--20
  (2014)

\bibitem{NIIBT15}
Nishimoto, T., I, T., Inenaga, S., Bannai, H., Takeda, M.: Dynamic index, {LZ}
  factorization, and {LCE} queries in compressed space. CoRR  abs/1504.06954
  (2015)

\bibitem{OS07}
Okanohara, D., Sadakane, K.: Practical entropy-compressed rank/select
  dictionary. In: Proc. 9th Workshop on Algorithm Engineering and Experiments
  (ALENEX). pp. 60--70 (2007)

\bibitem{Ryt03}
Rytter, W.: Application of {L}empel-{Z}iv factorization to the approximation of
  grammar-based compression. Theoretical Computer Science  302(1-3),  211--222
  (2003)

\bibitem{Sak05}
Sakamoto, H.: A fully linear-time approximation algorithm for grammar-based
  compression. Journal of Discrete Algorithms  3(2–4),  416--430 (2005)

\bibitem{ZL77}
Ziv, J., Lempel, A.: A universal algorithm for sequential data compression.
  IEEE Transactions on Information Theory  23(3),  337--343 (1977)

\end{thebibliography}

\end{document}